\documentclass[letterpaper, 10 pt, conference]{ieeeconf}  

\usepackage[noadjust]{cite}

\usepackage{ amssymb }
\usepackage{float}
\usepackage{amsmath}
\usepackage{amsfonts}
\usepackage{enumerate}  
\usepackage{graphicx}
\usepackage{bm}   
\usepackage{ dsfont }
\usepackage{mathtools}
\usepackage{upgreek}
\mathtoolsset{showonlyrefs=true}
\usepackage{etoolbox}
\usepackage{courier}
\usepackage{url}
\usepackage{xcolor}

\newcommand{\mfs}[1]{{{\textsf{\footnotesize #1}}}}

\newcommand{\mf}{\mathbf}
\newcommand{\sgn}[2]{\left\lceil#1\right\rfloor^{#2}}

\newtheorem{proposition}{Proposition}

\newtheorem{assumption}{Assumption}

\newtheorem{theorem}{Theorem}
\newtheorem{remark}{Remark}
\newtheorem{corollary}{Corollary}
\newtheorem{definition}{Definition}

\newcommand\blfootnote[1]{%
  \begingroup
  \renewcommand\thefootnote{}\footnote{#1}%
  \addtocounter{footnote}{-1}%
  \endgroup
}

\def\RevOne#1{{#1}}
\def\RevSix#1{{#1}}
\def\RevSeven#1{{#1}}
\def\RevAll#1{{#1}}

\title{
Dynamic consensus with prescribed convergence time for multi-leader formation tracking
}

\author{Rodrigo Aldana-López, David~Gómez-Gutiérrez, Rosario~Aragüés, and Carlos Sagüés.}

\begin{document}

\maketitle
\thispagestyle{empty}
\pagestyle{empty}

\begin{abstract}
This work addresses the problem of distributed formation tracking for a group of follower \RevSix{holonomic mobile robots} around a reference signal. The reference signal is comprised of the geometric center of the positions of multiple leaders. \RevAll{This work's main contribution is a novel Modulated Distributed Virtual Observer (MDVO) for the reference signal.  Moreover, the proposed MDVO is based on an exact dynamic consensus algorithm with a prescribed convergence time.} In addition, we provide simulation examples showcasing two different application scenarios for the proposal.
\blfootnote{This work was supported by projects COMMANDIA SOE2/P1/F0638
(Interreg Sudoe Programme, ERDF), PGC2018-098719-B-I00 (MCIU/ AEI/
FEDER, UE) and DGA T45-20R (Gobierno de Aragón). The authors would like to acknowledge the sponsorship of Universidad de Zaragoza, Banco Santander and CONACYT, Mexico.}
\blfootnote{Rodrigo Aldana-López, Rosario Aragüés and Carlos Sagüés are with Departamento de Informática e Ingeniería de Sistemas (DIIS) and Instituto de Investigación en Ingeniería de Aragón (I3A), 
Universidad de Zaragoza, Zaragoza 50018, Spain. David~Gómez-Gutiérrez is with Intel Labs, Intel Tecnología de México, Av. del Bosque 1001, 45019, Zapopan, Jalisco, Mexico and Instituto Tecnológico José Mario Molina Pasquel y Henríquez, Unidad Académica Zapopan. Cam. Arenero 1101, 45019, Zapopan, Jalisco, Mexico.
(e-mail:  {\tt\small rodrigo.aldana.lopez@gmail.com, raragues@unizar.es, csagues@unizar.es, david.gomez.g@ieee.org})}

\blfootnote{\textcolor{red}{This is the accepted version of the manuscript: R. Aldana-López, D. Gómez-Gutiérrez, R. Aragüés and C. Sagüés, ``Dynamic consensus with prescribed convergence time for multi-leader formation tracking," in IEEE Control Systems Letters, 2022, DOI: 10.1109/LCSYS.2022.3181784. 
\textbf{Please cite the publisher's version}. For the publisher's version and full citation details see:
\url{https://doi.org/10.1109/LCSYS.2022.3181784}. 
}
}
\blfootnote{``© 2022 IEEE.  Personal use of this material is permitted.  Permission from IEEE must be obtained for all other uses, in any current or future media, including reprinting/republishing this material for advertising or promotional purposes, creating new collective works, for resale or redistribution to servers or lists, or reuse of any copyrighted component of this work in other works.”}
\end{abstract}
\begin{keywords}
Distributed control; Formation control; Dynamic consensus; Leader-follower; Prescribed-time convergence
\end{keywords}

\section{Introduction}
Problems in the context of distributed Multi-Agent Systems (MAS) have persistently attracted attention in the control and robotics community during the last couple of decades. In particular, a popular tool for effectively achieving distributed collaboration along the MAS are consensus protocols. In this setting, the agents in the system agree on a quantity of interest by sharing information to neighbors in the network. The problem can be classified depending on whether the consensus value is static \cite{Olfati2007}, dynamic \cite{edcho} or when it is dictated by a leader \cite{Jadbabaie2003}. 

\RevAll{In addition, multi-leader distributed problems have attracted attention in the literature. For example, in the context of containment control and herding, there is interest in the situation in which follower robots maintain a formation inside the convex hull of the positions of a set of leader robots \cite{zeno}. This is useful in situations with heterogeneous robots, where leaders are equipped with sensors to detect obstacles, whereas followers are not. In this case, computing the geometric center of the leaders' positions and maintaining a formation around it can be used as a safe strategy for the follower robots. On the other hand, in the context of automated surveillance and escorting \cite{escort}, leader robots may be able to obtain imperfect detections of a single target for follower robots to achieve a formation around it. In this case, the geometric center of the leader estimations can be used as an improved estimation for the target. Hence, in this work, we are interested in the problem of distributed formation tracking of a group of follower \RevSix{holonomic mobile robots} around the geometric center of a group of leaders. }

A standard approach for the leader-follower tracking problem is to adopt an observer-based approach \cite{chuan2021}. In this case, a Distributed Virtual Observer (DVO) is used for the reference of interest, in this case, the geometric center of the leaders' positions. In a second step, a local controller is designed to drive the followers' state to the reference given by the DVO. This strategy has been used for second-order systems \cite{trujillo2020}, Euler-Lagrange systems \cite{cai2016}, and high order dynamics \cite{chuan2019}.

In the robotics context, it may be required for the DVO to obtain the global reference in real-time, i.e., before a time deadline \cite{trujillo2020}. This may be either mandated by the application or to ensure the correctness of the local controller if both are executed simultaneously. Finite-time stability concepts in control theory can be used to ensure convergence before a given deadline. \RevOne{In contrast to usual asymptotic convergence, fixed-time and prescribed-time stable systems have a uniformly bounded settling-time \cite{poly2012} which can be used to prescribe the convergence time. }\RevAll{A time-independent fixed-time approach has been used for leader-follower algorithms, e.g., in~\cite{defoort2015, trujillo2020, fx1, fx2, fx3}}. One of the disadvantages of these techniques is that an explicit expression for the settling-time bound as a function of initial conditions is usually unknown, limiting its applicability for real-time. In addition, even if a settling-time bound is known, it is usually overestimated, leading to over-engineered solutions.

Other works use persistently growing time-varying gains, which ultimately become singular at the deadline. \RevAll{This technique was used for consensus in \cite{ning2019,zhao2019, holeader, holeaderevent,trujillo2020}}. The most significant disadvantage of this approach is that a small disturbance easily compromises the system's behavior  due to the high gains employed. Moreover, Time Base Generators (TBG) were used as tracking references in \cite{tbg} achieving convergence before a deadline in a single-leader scenario with agents of high order dynamics. A similar multi-leader problem is tackled in \cite{zeno}. \RevAll{However, convergence in prescribed-time is obtained by means of a switching algorithm inducing Zeno behaviour.}

Motivated by the previous discussion, we propose a framework for multi-leader formation tracking. \RevAll{The main contribution of the proposal is a Modulated Distributed Virtual Observer (MDVO) which  computes the geometric center of the multiple leaders' position. One of the building-blocks of the MDVO is an adaptation for real-time constraints of the EDCHO protocol \cite{edcho} \RevAll{which was shown to have robust features compared to linear and first order sliding mode protocols as discussed in \cite{edcho, redcho}. The new protocol introduces TBGs as modulating functions in order to prescribe the convergence time, where we provide a formal stability analysis.}}  Furthermore, the versatility of MDVO is shown through simulations showcasing two different application scenarios for multi-robot systems. \RevAll{The differences of MDVO with respect to the rest of the literature are also discussed}.\footnote{{Simulation files for the algorithms presented in this work can be found on \url{https://github.com/RodrigoAldana/EDC}.}}

\subsection{Notation}
Let $\mathcal{C}^{m+1}[0,\infty)$ the set of all functions $f:[0,\infty)\to\mathbb{R}$ from which the $(m+1)$-th derivative $f^{(m+1)}(t)$ is continuous $\forall t\geq 0$. Given a vector \RevOne{$\mf{x}=[x_1,\dots,x_n]^T$ denote the Euclidean norm as $\|\mf{x}\|$ and the $\infty$-norm as $\|\mf{x}\|_\infty:=\max_{1\leq i\leq n} |x_i|$}. $\binom{\mu}{\nu}$ denotes the binomial coefficient. Let $\text{sign}(x) = 1$ if $x> 0, \text{sign}(x)=-1$ if $x<0$ and $\text{sign}(0)=0$. Moreover, if $x\in\mathbb{R}$, let $\lceil x\rfloor^\alpha:= |x|^\alpha\text{sign}(x)$ for $\alpha>0$ and $\lceil x\rfloor^0:={\text{sign}}(x)$. When $\mf{x}=[x_1,\dots,x_n]^T\in\mathbb{R}^n$, then $\sgn{\mf{x}}{\alpha}:=\left[\sgn{x_1}{\alpha},\dots,\sgn{x_n}{\alpha}\right]^T$ for $\alpha\geq 0$.
\section{Problem statement}
\label{sec:problem}
Consider a multi-agent system of $\mfs{N}$ \RevSix{holonomic mobile robots}. From this team of robots, $\mfs{N}_{\mfs{L}}\geq 1$ robots are considered to be leaders whereas the remaining $\mfs{N}_{\mfs{F}} = \mfs{N}-\mfs{N}_{\mfs{L}}$ are followers. Moreover, all robots are connected through a communication network modeled by an undirected graph $\mathcal{G}$, where a node corresponds to a single robot and an edge corresponds to a bi-directional communication link between two robots. We do not assume that all followers are directly connected to a leader, but that $\mathcal{G}$ is connected. In addition, for convenience in the presentation, we index the robots using the index sets $\mathcal{I}_{\mfs{L}} = \{1,\dots,\mfs{N}_{\mfs{L}}\}$  for the leaders, $\mathcal{I}_{\mfs{F}} = \{\mfs{N}_{\mfs{L}}+1,\dots,\mfs{N}\}$ for the followers and $\mathcal{I}=\mathcal{I}_{\mfs{L}}\cup\mathcal{I}_{\mfs{F}}$. Let $\mf{p}_i(t)\in\mathbb{R}^3$ be the position of the $i$-th robot in $\mathbb{R}^3$. For simplicity, we assume that the follower robots have dynamics
\begin{equation}
    \label{eq:dynamics}
    \mf{p}_i^{(m)}(t) = \mf{u}_i(t), i\in\mathcal{I}_{\mfs{F}}
\end{equation}
where $ \mf{p}_i^{(m)}(t)$ is the $m$-th derivative of $ \mf{p}_i(t), i\in\mathcal{I}_{\mfs{F}}$ for given $m\in\mathbb{N}$ and $\mf{u}_i(t)\in\mathbb{R}^3, i\in\mathcal{I}_{\mfs{F}}$ is a local control input driving the robot dynamics. Despite the simplicity of the model in \eqref{eq:dynamics}, the results presented in this work are easily adapted to different linear dynamics.  On the other hand, we assume that the leaders have arbitrary $(m+1)$-times differentiable trajectories $\mf{p}_i(t), \forall i\in\mathcal{I}_{\mfs{L}}$. However, the leaders can cooperate into communicating information to the network or collaborating in some form of distributed computation.  The goal of the followers is to choose $\mf{u}_i(t),i\in\mathcal{I}_{\mfs{F}}$ in order to achieve the following:
\begin{definition} 
\label{def:ml_tracking}
We say that the followers achieved multi-leader formation tracking if $
    \lim_{t\to\infty}\left\|\mf{p}_i(t) - \mf{\bar{p}}(t) - \mf{d}_i\right\| = 0, \forall i\in\mathcal{I}_\mfs{F}$ where 
\begin{equation}
\label{eq:correct_average}
\mf{\bar{p}}(t):=\frac{\mf{p}_1(t)+\cdots+\mf{p}_{\mfs{N}_\mfs{L}}(t)}{\mfs{N}_{\mfs{L}}}
\end{equation} for formation displacements $\mf{d}_i\in\mathbb{R}^3$ \RevSeven{fixed beforehand}.
\end{definition}

In order to perform this task, a Distributed Virtual Observer (DVO) estimating the reference $\mf{\bar{p}}(t)$ can be used at each robot so that the controller $\mf{u}_i(t), i\in\mathcal{I}_\mfs{F}$ is designed for $\mf{p}_i(t)$ to track such reference. \RevAll{The desired properties of such DVO are enlisted as follows.}
\begin{definition}
\label{def:rt_dvo}
\RevAll{A Distributed Virtual Observer (DVO) for $\mf{\bar{p}}(t)$ is an algorithm that runs at each robot $i\in\mathcal{I}$ and computes local estimations $\hat{\mf{p}}_{i,0}(t),\dots,\hat{\mf{p}}_{i,m}(t)\in\mathbb{R}^3$ for $\mf{\bar{p}}(t)$ and its first $m$ derivatives respectively. The DVO at robot $i\in\mathcal{I}$ communicates ${\mfs{N}_\mfs{a}}>0$ reals in a vector $\mf{a}_i(t)\in\mathbb{R}^{\mfs{N}_\mfs{a}}$ only to its neighbors through the network $\mathcal{G}$.} Moreover, a DVO has a prescribed convergence time \RevSix{\cite{tbg,trujillo2020}} if given $T_c>0$ then $\mf{\hat{p}}_{i,\mu}(t) = \mf{\bar{p}}^{(\mu)}(t), \forall t\geq T_c, \forall i\in\mathcal{I}, \forall \mu\in\{0,\dots,m\}$.
\end{definition}

\RevAll{To solve this problem we propose a Modulated Distributed Virtual Observer (MDVO), outlined as follows.} First, we assign a \RevSeven{local} vector signal $\mf{s}_i(t)\in\mathbb{R}^3$ at each robot depending on if it is a leader or a follower:
\begin{equation}
\label{eq:dyn_signals}
    \mf{s}_i(t) = \left\{
    \begin{array}{cl}
    \mf{p}_i(t) & \text{ if } i\in\mathcal{I}_{\mfs{L}} \\
    0 & \text{ if } i\in\mathcal{I}_{\mfs{F}} \\
    \end{array}
    \right.
\end{equation}
\RevSeven{To allow a more general setting, we do not assume the global number of leaders $\mfs{N}_\mfs{L}$ to be known to all robots}. Instead, we assign auxiliary scalar \RevSeven{local} labels \RevAll{$\ell_i$} as
\begin{equation}
\label{eq:stat_signals}
    \RevAll{\ell_i} = \left\{
    \begin{array}{cl}
    1 & \text{ if } i\in\mathcal{I}_{\mfs{L}} \\
    0 & \text{ if } i\in\mathcal{I}_{\mfs{F}} \\
    \end{array}
    \right.
\end{equation}
Then, MDVO uses exact dynamic consensus tools in order to compute the following averages through the $\mfs{N}$ robots in a distributed fashion:
\begin{equation}
\begin{aligned}
    \label{eq:averages}
    \mf{\bar{s}}(t)&= \frac{\mf{s}_1(t)+\dots+\mf{s}_{\mfs{N}}(t)}{\mfs{N}}=\frac{\mf{p}_1(t)+\dots+\mf{p}_{\mfs{N}_{\mfs{L}}}(t)}{\mfs{N}}\\
    \bar{\ell} &=\frac{\ell_1+\dots+\ell_{\mfs{N}}}{\mfs{N}} =\frac{\mfs{N}_{\mfs{L}}}{\mfs{N}}
\end{aligned}
\end{equation}

Note that $\bar{\mf{p}}(t)$ can be computed from the ratio $\mf{\bar{s}}(t)/\bar{\ell}$. Once each robot has access to the signal $\bar{\mf{p}}(t)$, a trajectory tracking controller can be applied. Note that computing the averages \eqref{eq:averages} with prescribed convergence time is challenging due to the time-varying character of the signals $\mf{s}_i(t), i\in\mathcal{I}$. \RevAll{Hence, MDVO includes 4 distributed dynamic consensus blocks, one to compute each component of $\bar{\mf{s}}(t)$ and $\bar{\ell}$. Although we do not assume knowledge of neither $\mfs{N}$ nor $\mfs{N}_{\mfs{L}}$, we assume that all robots run the MDVO algorithm with the same parameters, fixed beforehand. \RevSeven{Moreover, each robot $i\in\mathcal{I}$ has access to its own signal $\mf{s}_i(t)$ and label $\ell_i$}. This allows the possibility of the same algorithm to work for different network configurations and sizes. Moreover, as described in Section \ref{sec:radvo}, each robot $i\in\mathcal{I}$ shares $\mfs{N}_{\mfs{a}}=4$ numbers, one for each consensus block, to their neighbors.}

\section{\RevOne{Modulated Dynamic Consensus}}
\label{sec:medcho}
In this section, we present the general form of the distributed dynamic consensus block used in the proposed MDVO. Here, the block is presented for arbitrary scalar signals $s_i(t)$ located at each robot. The protocol has $m+1$ internal states $x_{i,0}(t),\dots,x_{i,m}(t)$ and virtual outputs $y_{i,0}(t),\dots,y_{i,m}(t)$ at each agent $i\in\mathcal{I}$. Given a deadline $T_c>0$, the purpose of the protocol is that each agent achieves
$
y_{i,\mu}(t) = \bar{s}^{(\mu)}(t), \ \ \ \forall t\geq T_c, \forall i\in\mathcal{I}, \forall \mu\in\{0,\dots,m\}
$
where 
$
\bar{s}(t):= (1/\mfs{N})\sum_{i=1}^{\mfs{N}}s_i(t)
$. The algorithm is written as:
\begin{equation}
\label{eq:medcho}
\begin{aligned}
    \dot{x}_{i,\mu} &=  \theta^{\frac{\mu+1}{m+1}}k_\mu \mbox{$\sum_{j=1}^{\mfs{N}}$}a_{ij}\lceil y_{i,0}- y_{j,0} \rfloor^{\frac{m-\mu}{m+1}} + x_{i,\mu+1} \\
    &\text{for }0\leq\mu < m\\
    \dot{x}_{i,m} &= \theta k_m \mbox{$\sum_{j=1}^{\mfs{N}}$}a_{ij}\sgn{y_{i,0} - y_{j,0} }{0} \\
    y_{i,\mu} &= \sigma_i^{(\mu)} - x_{i,\mu}.
\end{aligned}
\end{equation}
where time dependence was omitted for brevity, $a_{ij}\in\{0,1\}$ are the components of the adjacency matrix of $\mathcal{G}$ \RevSix{such that only neighbors information contribute to \eqref{eq:medcho}. } Moreover, 
\begin{equation}
\label{eq:sigma}
\sigma_i^{(\mu)}(t) := \dfrac{\text{d}^\mu}{\text{d}t^{\mu}}\bigg(\kappa\left(\frac{t}{T_c}\right)s_i(t)\bigg) 
\end{equation}
with the introduction of the function $\kappa(\bullet)$ to modulate the signals $s_i(t)$ in order to meet deadline at $T_c$. This function must comply with the following properties:
\begin{definition}
\label{def:modulating}
A $\mathcal{C}^{m+1}[0,\infty)$ function $\kappa(\bullet)$ is called an $m$-th order modulating function if $\kappa(0) = 0$, $\kappa(t) = 1, \forall t\geq 1$ and $\kappa^{(\mu)}(0) = \kappa^{(\mu)}(1) = 0, \forall t\geq 1, \forall \mu\in\{1,\dots,m\}$.
\end{definition}
An example of such function is provided in Appendix \ref{ap:modulating}. 

\RevAll{Note for a single \eqref{eq:medcho} block, agents only share their output $y_{i,0}(t)$.} \RevSix{Moreover, \eqref{eq:medcho} uses the parameter $\theta$ and the  modulating function $\kappa(\bullet)$ to start at consensus with $y_{i,\mu}(0)=0$ and guide the convergence of the algorithm outputs towards $y_{i,\mu}(T_c)=\bar{s}^{(\mu)}(T_c)$ as shown in the following:}
\begin{assumption}
\label{as:all_bounded}
Given bounds $L_0,\dots,L_{m+1}>0$, then $\forall \mu\in\{0,\dots,m+1\}$ it follows that $\left|\bar{s}^{(\mu)}(t)-s_i^{(\mu)}(t)\right|\leq L_\mu, \forall t\in [0,T_c], \left|\bar{s}^{(m+1)}(t)-s_i^{(m+1)}(t)\right|\leq L_{m+1}, \forall t\geq T_c$.
\end{assumption}
\begin{theorem}
\label{th:medcho}
Let Assumption \ref{as:all_bounded} \RevSeven{hold}, $\mathcal{G}$ be a fixed connected graph and $k_0,\dots,k_m>0$ as in Proposition \ref{prop:edcho} in Appendix \ref{ap:edcho} for $L=1$. Moreover, let $x_{i,\mu}(0) = 0, \forall i\in\mathcal{I}, \forall\mu\in\{0,\dots,m\}$ and $\kappa(t)$ be a modulating function of order $m$. Denote with $K_{\mu}:=\sup_{t\in[0,1]}|\kappa^{(\mu)}(t)|$ and choose
\begin{equation}
\label{eq:gain}
    \theta = \sum_{\nu=0}^{m+1}\binom{m+1}{\nu} \frac{1}{T_{\min}^{m-\nu+1}}K_{m-\nu+1}L_\nu
\end{equation}
for some $T_{\min}>0$. Hence, given any deadline $T_c\geq T_{\min}$, the protocol \eqref{eq:medcho} achieves $y_{i,\mu}(t)=\bar{s}^{(\mu)}(t), \forall t\geq T_c, \forall i\in\mathcal{I}, \forall \mu\in\{0,\dots,m\}$.
\end{theorem}
\begin{proof}
The proof can be found in Appendix \ref{ap:medcho} 
\end{proof}
\begin{remark}
Note that from \eqref{eq:gain}, as the minimum allowed deadline $T_{\min}$ is decreased, the gain $\theta$ increases. This is expected since a smaller deadline will require greater correction effort for faster results in the protocol. 
\end{remark}
\begin{remark}
\label{rem:redcho}
Note that the proposed protocol assumes that all agents have knowledge of a synchronized global timer $t$ to compute the modulating function $\kappa(t)$. This assumption has been used extensively in other works (see \cite{ning2019,zhao2019,trujillo2020}), which is particularly important when a notion of a deadline $T_c>0$ is present. \RevSeven{However, it can be shown that due to finite-time stability character of \eqref{eq:medcho}, a small deviation in the timers for the signal $\kappa(t)$ leads to at most small changes in the settling time as well. }

\end{remark}

\section{Multi-leader formation tracking}

\subsection{\RevOne{Modulated Distributed Virtual Observer}}
\label{sec:radvo}

Equipped with the protocol in \eqref{eq:medcho}, MDVO is constructed using a ratio-consensus strategy. Recall the multi-leader follower setting from Section \ref{sec:problem} and the local signals $\mf{s}_i(t), \ell_i$ from \eqref{eq:dyn_signals} and $\eqref{eq:stat_signals}$. Moreover,  write $\mf{s}_i(t)$ through its components as $\mf{s}_i(t) = [s_i^{\mfs{X}}(t),s_i^{\mfs{Y}}(t),s_i^{\mfs{Z}}(t)]^T\in\mathbb{R}^3$. MDVO has four different instances of \eqref{eq:medcho}. In particular, denote with $\mathcal{M}_{\mf{s}}$ three identical \eqref{eq:medcho} blocks, one for each $s_i^{\mfs{X}}(t),s_i^{\mfs{Y}}(t),s_i^{\mfs{Z}}(t)$. Let ${\mf{y}}_{i,\mu}(t)\in\mathbb{R}^3$ a vector containing the three different $\mu$-th outputs at robot $i$ for $\mathcal{M}_{\mf{s}}$. Similarly, denote with $\mathcal{M}_{\ell}$ the block \eqref{eq:medcho} applied to labels $\ell_i$ with $\mu$-th output $l_{i,\mu}(t)$. \RevSeven{Hence, each robot $i\in\mathcal{I}$ shares $[\mf{y}_{i,0}(t)^T, l_{i,0}(t)]^T\in\mathbb{R}^{\mfs{N}_{\mf{a}}}, \mfs{N}_{\mf{a}}=4$}. The outputs of MDVO are
\begin{equation}
\label{eq:ratio}
\hat{\mf{p}}_{i,\mu}(t) := \frac{\mf{y}_{i,\mu}(t)}{\max\left( l_{i,0}(t), 1/\mfs{N}_{\max}\right)}, i\in\mathcal{I}, \mu\in\{0,\dots,m\}
\end{equation}
with a maximum allowed number of nodes $\mfs{N}_{\max}\geq \mfs{N}$ in the network known by all robots. In order to ensure convergence of $\mathcal{M}_{\mf{s}}$, the leaders $\mf{p}_{i}(t), i\in\mathcal{I}_{{\mfs{L}}}$ cooperate such that the following is complied.
\begin{assumption}
\label{as:motion}
Given known $L_0,\dots,L_{m+1}>0$, the motion of any leader robot $i\in\mathcal{I}_{\mfs{L}}$ is constrained as
$
\|\mf{p}_i^{(\mu)}(t)\|_\infty\leq L_\mu, \forall \mu\in\{0,\dots,m+1\}
$
with $t\in[0,T_c]$ and
$
\|\mf{p}_i^{(m+1)}(t)\|_\infty\leq L_{m+1}$ with $t\geq T_c$.
\end{assumption}
\begin{theorem}
\label{th:ratio}
Let Assumption \ref{as:motion} \RevSeven{hold}, $\mathcal{G}$ be a fixed connected graph and a maximum allowed number of nodes $\mfs{N}_{\max}\geq \mfs{N}$ in the network known by all robots.
Moreover, let the instances of $\mathcal{M}_{\mf{s}}$ be designed as in Theorem \ref{th:medcho} for $L_0,\dots,L_{m+1}$ from Assumption \ref{as:motion}. Let the instance of $\mathcal{M}_{\mf{\ell}}$ designed as $\mathcal{M}_{\mf{s}}$ except that 
$\theta = \frac{K_{m+1}}{T_{\min}^{m-\nu+1}}$
is chosen instead. Hence, this configuration along with outputs $\hat{\mf{p}}_{i,\mu}(t), \mu\in\{0,\dots,m\}$ as in \eqref{eq:ratio}
is a DVO for $\mf{\bar{p}}(t)$ with prescribed convergence time $T_c$ in the sense of Definition \ref{def:rt_dvo}.
\end{theorem}
\begin{proof}
The proof can be found in Appendix \ref{ap:ratio}.
\end{proof}

\begin{remark}
\label{rem:nmax}
Theorem \ref{th:ratio} assumes knowledge of $\mfs{N}_{\max}$ which \RevSix{can be overestimated by anything arbitrarily bigger than $\mfs{N}$}. This avoids any division by zero in \eqref{eq:ratio}. This assumption is not restrictive in many situations where the communication link between robots only allow a fixed maximum number of identification numbers for each node e.g., at the data-link layer. However, if a value of $\mfs{N}_{\max}$ is unavailable, \eqref{eq:ratio} may be computed directly as $\hat{\mf{p}}_{i,\mu}(t) = \mf{y}_{i,\mu}(t)/l_{i,0}(t)$ for $t\geq T_c$.
\end{remark}
\begin{remark}
Note from \eqref{eq:medcho} that if for some $i\in\mathcal{I}$, $s_i(t)= 0, \forall t\geq 0$, the output comply $y_{i,\mu}(t) = s^{(\mu)}_i(t)-x_{i,\mu}(t)=-x_{i,\mu}(t)$ leading to $x_{i,\mu}(t) = -\bar{s}^{(\mu)}(t), \forall t\geq T_c$ without the need of any additional differentiator to compute $s^{(\mu)}_i(t)$. This is the case of the ratio-consensus strategy for all followers $i\in\mathcal{I}_{\mfs{F}}$ where $s_i^{\mfs{X}}(t) = s_i^{\mfs{Y}}(t) = s_i^{\mfs{Z}}(t) = \ell_i = 0, \forall t\geq 0$.
\end{remark}

\subsection{Formation tracking}
\label{sec:control}
Access to the signal $\mf{\bar{p}}(t)$ and its first $m$ derivatives through outputs $\hat{\mf{p}}_{i,0}(t),\dots ,\hat{\mf{p}}_{i,m}(t), i\in\mathcal{I}$ of the MDVO enables the design of local controllers $\mf{u}_i(t), i\in\mathcal{I}_{\mfs{F}}$ in \eqref{eq:dynamics} such that multi-leader formation tracking is achieved. Recall the fixed formation positions $\mf{d}_i,i\in\mathcal{I}_{\mfs{F}}$ such that the proposed control design has the form
\begin{equation}
\label{eq:controller}
\begin{aligned}
    \mf{u}_i(t) &= \hat{\mf{p}}_{i,m}(t) - \rho_0(\mf{p}_i(t) - \mf{\hat{p}}_{i,0}(t) - \mf{d}_i)  \\ 
    &- \sum_{\mu=1}^{m-1} \rho_\mu( \mf{p}_i^{(\mu)}(t) -  \mf{\hat{p}}_{i,\mu}(t))
    \end{aligned}
\end{equation}
for $i\in\mathcal{I}_{\mfs{F}}$, where $\rho_0,\dots,\rho_{m-1}>0$ are chosen such that  $\lambda^{m}+\sum_{\mu=0}^{m-1}\rho_{\mu-1}\lambda^{\mu-1}
$ have roots in the left half plane. 
\begin{corollary}
\label{cor:controller}
Let MDVO designed as in Theorem \ref{th:ratio} and Assumption \ref{as:motion} hold. Then, \eqref{eq:controller} makes the followers achieve multi-leader formation tracking as in Definition \ref{def:ml_tracking}.
\end{corollary}
\begin{proof}
\RevSix{First, note that each instance of \eqref{eq:medcho} in the MDVO is finite-time stable as a result of Theorem \ref{th:medcho}, thus having bounded outputs for $t\in[0,T_c]$. As a result, the output \eqref{eq:ratio} of MDVO is bounded for all $t\in[0,T_c]$ as well. Hence, the linear system \eqref{eq:dynamics} being input to state stable does not exhibit any finite-time escape before $t=T_c$. Now, let $t\geq T_c$ and note that the separation principle is complied since Theorem \ref{th:ratio} implies $\hat{\mf{p}}_{i,m}(t)=\mf{\bar{p}}^{(\mu)}(t), \forall t\geq T_c$} with $\mf{\bar{p}}(t)$ from \eqref{eq:correct_average}. Set $\mf{e}_i(t)=\mf{p}_i(t) - \mf{\bar{p}}(t)- \mf{d}_i$ with closed-loop error dynamics for $i\in\mathcal{I}_{\mfs{F}}$ complying
$
\mf{e}_i^{(m)}(t) = -\sum_{\mu=0}^{m-1} \rho_\mu \mf{e}_i^{(\mu)}(t)
$
for $t\geq T_c$, being asymptotically stable towards the origin. Thus, $\left\|\mf{p}_i(t) - \mf{\bar{p}}(t) - \mf{d}_i\right\|\to 0$ as $t\to \infty$.
\end{proof}

\begin{remark}
\label{rem:extension}
In this work, a linear trajectory tracking controller \eqref{eq:controller} is considered. However, it is important to note that once each agent is equipped with the value of the reference $\mf{\bar{p}}(t)$ along with its derivatives as computed from the MDVO, the employment of a wide variety of trajectory tracking controllers for integrator systems is enabled. In particular, it enables the usage of controllers which incorporate barrier functions for collision avoidance as in \cite{panagou} \RevOne{or to reject disturbances as in \cite{tbg}}.
\end{remark}

\section{Simulation examples}
\label{sec:examples}
In this section, we assume all robots of order $m=3$. Consensus blocks from \eqref{eq:medcho} are configured with $m=3$, gains $k_0=7.5, k_1=19.25, k_2=17.75, k_3=7$ taken from \cite{edcho} and the modulating function in \eqref{eq:interpolator} from Appendix \ref{ap:modulating} for $m=3$. Similarly, all robots use the controller \eqref{eq:controller} \RevSix{for all $t\geq 0$}. Consider $T_{\min}=T_c=0.5$, \RevSix{$\mfs{N}=8$ and a very overestimated $N_{\max}=10^3\geq \mfs{N}$}. Simulations where performed using the explicit Euler method with time step $\Delta t=10^{-6}$.

\subsection{Herding with cooperative sheep}
\label{sec:shepherd}
Consider a team of $\mfs{N}_{\mfs{L}}=3$ shepherds and $\mfs{N}_{\mfs{F}}=5$ sheep robots with initial positions and communication topology shown in Figure \ref{fig:sheep_RADVO}-a). The goal is for the shepherds to move along a trajectory, dragging the sheep robots along a fixed circular formation around the geometric center of the shepherd positions. Figure \ref{fig:sheep_RADVO}-b) shows the signals $l_{i,0}(t)$, computed in the MDVO for all robots. The signals start at consensus and converge to $\mfs{N}_{\mfs{L}}/\mfs{N} = 3/8$ before the deadline $T_c=0.5$. Similarly, as shown in Figure \ref{fig:sheep_RADVO}-b) the signals $\mf{\hat{p}}_{i,0}(t)$ converge to $\mf{\bar{p}}(t)$ exactly at the deadline $T_c=0.5$ as well. The closed loop formation error dynamics when using the controller \eqref{eq:controller} converges to the origin as depicted in Figure \ref{fig:sheep_RADVO}-d), such that the final formation is achieved.

\begin{figure}
    \centering
    \vspace{0.5em}
    \includegraphics[width=0.4\textwidth]{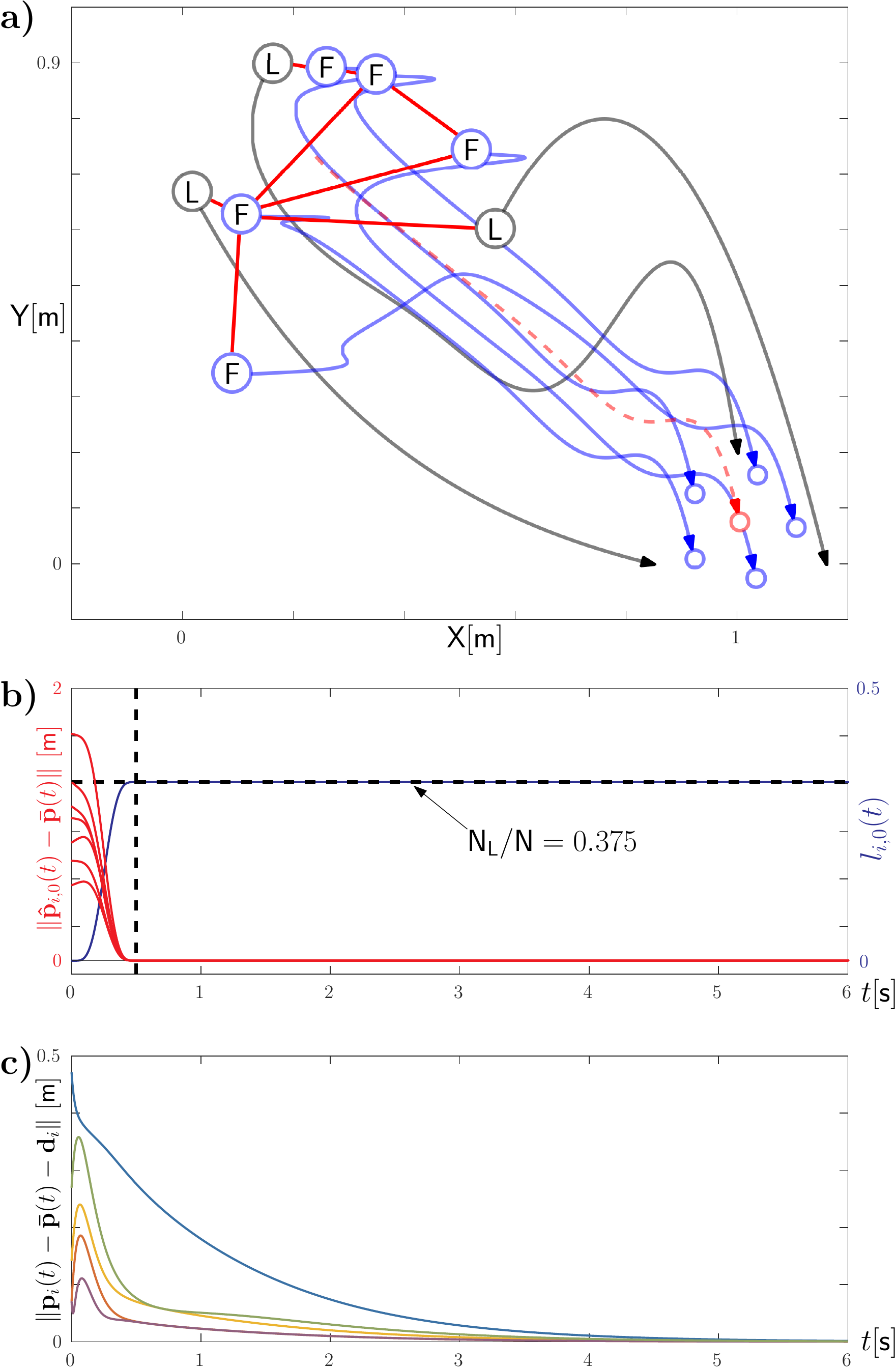}
    \caption{Experimental results for the formation control as in of Section \ref{sec:shepherd}. a) Nodes labeled with $\mfs{L}$ are leader robots, whereas $\mfs{F}$ are followers. Continuous lines show the trajectory in the $\mfs{XY}$ plane of all robots. Dotted red line shows the trajectory of the geometric center of the leaders $\bar{\mf{p}}(t)$. Continuous red lines show communication links between the robots. b) The signal $l_{i,0}(t)$ (blue) in the MDVO is shown for all agents, achieving the correct value of $\mfs{N}_{\mfs{L}}/\mfs{N}=3/8$. All agents maintain consensus from the start due to the modulating function. The dotted black line denotes the deadline at $T_c=0.5$. The consensus error is depicted (red), reaching the origin before the deadline. c) The formation error is shown, reaching the origin asymptotically and the followers achieving the formation as observed in a). }
    \label{fig:sheep_RADVO}
\end{figure}

\subsection{Cooperative target estimation with \RevOne{uncertain measurements}}
\label{sec:target}
Consider $\mfs{N}_{\mfs{L}}=3$ robots equipped with sensors capable of detecting a moving target of interest $\mf{p}_{\mfs{T}}(t)=[x_{\mfs{T}}(t),y_{\mfs{T}}(t),z_{\mfs{T}}(t)]^T$ with $m$-th order dynamics. On the other hand, consider $\mfs{N}_{\mfs{F}}=5$ auxiliary robots which cannot detect the target, but are required to follow a formation around it. In this case, the leaders instead of sharing their position, share their estimate of the target position as $\mf{p}_i(t)=\mf{p}_{\mfs{T}}(t)+\mf{n}_i(t)$ impregnated with noise $\mf{n}_i(t)$. The signal $\mf{n}_i(t)$ has components sampled from a Gaussian distribution with zero mean and variance 1 at each time step, modeling detection errors. Hence, all robots run the MDVO in order to obtain the averaged target estimate $\mf{\bar{p}}(t)$, which is an improved global version of the individual estimates $\mf{p}_i(t)$. For the sake of brevity, Figure \ref{fig:target} only shows the performance in the estimation of the signal $\mf{p}_{\mfs{T}}(t)$ for the $\mfs{X}$ coordinate. Note that the average error for the $\mfs{X}$ coordinate estimation at the leaders $|s^{\mfs{X}}_{i}(t)-x_{\mfs{T}}(t)|$ is $0.7986, 0.797, 0.746$ respectively. However, the average error for the signal $|\mf{\hat{p}}^{\mfs{X}}_{i,0}(t)-x_{\mfs{T}}(t)|$ is of $0.4769, 0.4676, 0.4677, 0.4630, 0.4763$ respectively at the followers \RevSeven{where a noise attenuation of $0.58\approx 1/\sqrt{\mfs{N}_{\mfs{L}}}=1/\sqrt{3}$ is observed} as expected from an averaged estimate.

\begin{figure}
    \centering
    \vspace{1em}
    \includegraphics[width=0.4\textwidth]{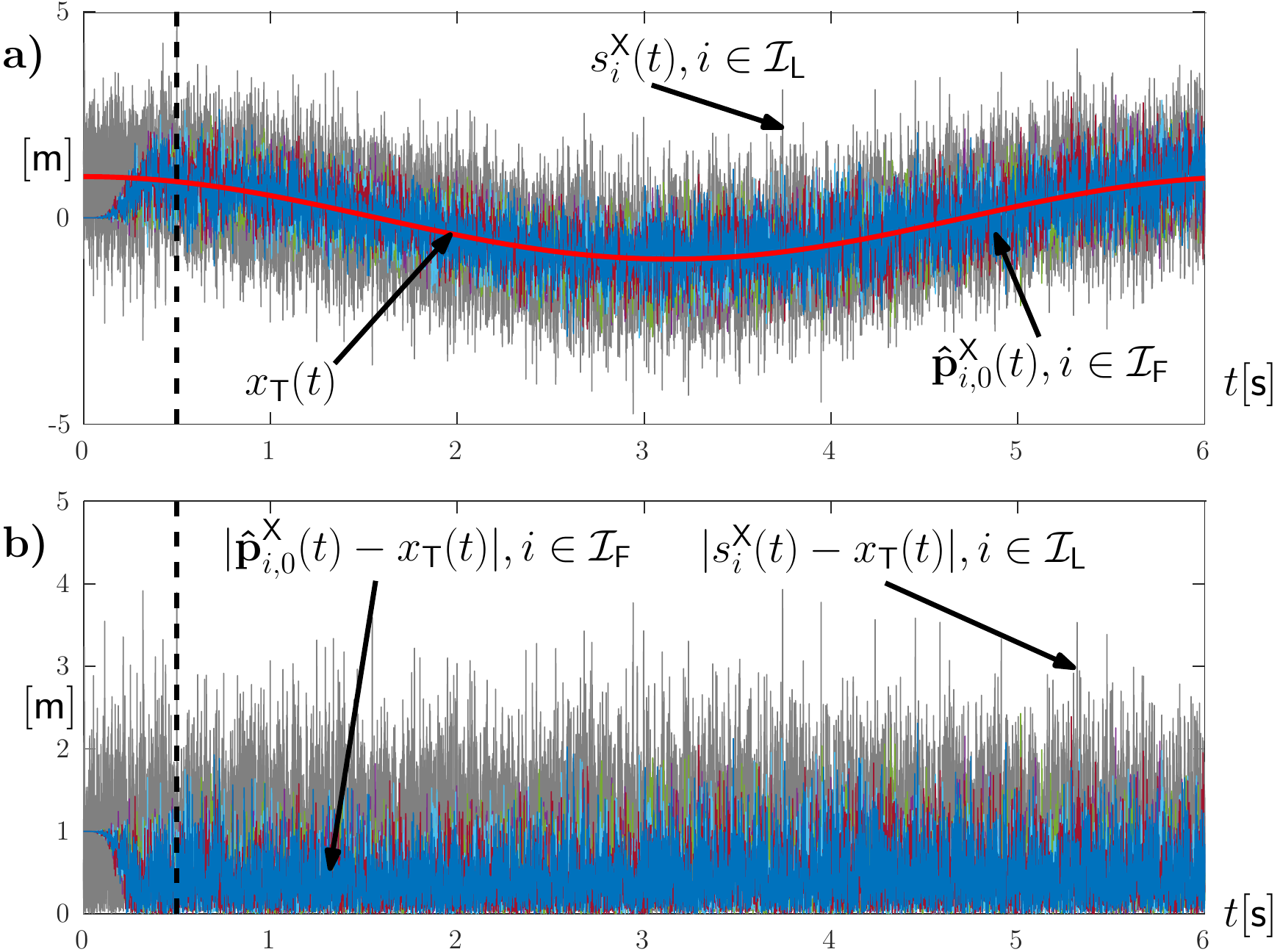}
    \caption{Experimental results for the cooperative target estimation as in Section \ref{sec:target}. a) Shows the target trajectory $x_\mfs{T}(t)$  in the $\mfs{X}$ axis (red), the $\mfs{N}_{\mfs{L}} = 3$ leader estimations $s_i^{\mfs{X}}(t)$ (grey) and the collective estimations $\mf{p}_{i,0}^{\mfs{X}}(t)$ (color) at the followers. b) Shows the estimation error for such signals, \RevSeven{where a noise attenuation of $0.58\approx 1/\sqrt{\mfs{N}_{\mfs{L}}}=1/\sqrt{3}$ is observed.}}
    \label{fig:target}
\end{figure}

\section{\RevAll{Comparison with related works}}
\label{sec:comparison}
\RevAll{For comparison, first note that the leader-follower strategies based in autonomous fixed-time consensus protocols in \cite{defoort2015, trujillo2020, fx1, fx2, fx3} cannot be extended to the multi-leader setting and have a greatly overestimated settling time estimation in the high order case \cite{fx1,fx2,fx3}, leading to over-engineered solutions in practice. In addition, it is usual for these approaches to require knowledge of global information such as the algebraic connectivity of the network in order to set the convergence time bound. In contrast, the proposed MDVO works in a multi-leader setting with a non-conservative settling time estimate due to the use of the modulating function in \eqref{eq:medcho} and as shown in Figure \ref{fig:sheep_RADVO}, without requiring knowledge of the number of leaders or the algebraic connectivity of $\mathcal{G}$.}

\RevAll{The works \cite{ning2019,zhao2019,trujillo2020, holeader, holeaderevent} use time-dependent prescribed-time consensus protocols that require persistently growing unbounded gains, problematic to implement in practice. Similarly, MDVO is based on a time-dependent protocol due to use of the modulating function. However, we do not require unbounded gains in \eqref{eq:medcho}. In \cite{zeno} a time-varying protocol is used for a multi-leader problem. However, different from MDVO, such protocol requires a sequence of sampling instants inducing Zeno behaviour and complicating its discrete-time implementation under a lower bounded time step.}

\RevAll{It is important to note that the example in Section \ref{sec:shepherd} was chosen precisely since it cannot be solved by existing approaches as a whole.  The reason is that other leader-follower approaches \cite{ning2019,zhao2019,defoort2015, trujillo2020, tbg, fx1, fx2, fx3, zeno, holeader, holeaderevent}, cannot be directly extended to handle the computation of the geometric center of multiple leaders, or to dynamics of order higher than 2. Moreover, note that the same observer structure of MDVO is used in Section \ref{sec:target} for a different problem than in Section \ref{sec:shepherd}, which highlights the versatility of our proposal in contrast to prior work. Another important improvement of MDVO is that the number of numbers shared between agents $\mfs{N}_\mfs{a}$ remains the same regardless of the system order $m$. In contrast, other high-order techniques \cite{fx1,fx2,fx3,tbg,zeno,holeader, holeaderevent} share a vector of a size growing with $m$.}

\section{Conclusion}
An MDVO was presented for its use in multi-leader formation tracking of \RevSix{holonomic mobile robots}. It was shown that the proposed MDVO, based in distributed dynamic consensus blocks, is able to compute the geometric center of the average of the leader positions in a prescribed convergence time. The versatility of the proposal was shown in different application scenarios for multi-leader formation control. \RevOne{This discussion motivate the formal analysis of the MDVO under communication delays in a future work.}

\appendix

\subsection{Examples of modulating functions}
\label{ap:modulating}

Given $m\in\mathbb{N}$, an $m$-th order modulating function $\kappa(t)$ can be constructed as follows. First, define $\bm{\upkappa}(t) = [\kappa^{(0)}(t),\dots, \kappa^{(m)}(t)]^T$. Moreover, let $A\in\mathbb{R}^{(m+1)\times (m+1)},B\in\mathbb{R}^{(m+1)\times 1}$ be the appropriate matrices modeling an $(m+1)$-th order integrator such that
\begin{equation}
\label{eq:kappa_sys}
\dot{\bm{\upkappa}}(t) = A{\bm{\upkappa}}(t) + B\kappa^{(m+1)}(t), \ \bm{\upkappa}(0)=0
\end{equation}
Furthermore, let $W =\int_0^1 \exp(A^T\tau)BB^T\exp(A\tau)\text{d}\tau$, ${\bm{\upkappa}}_f = [1,0,\dots,0]^T$ and 
\begin{equation}
\label{eq:interpolator}
    \kappa^{(m+1)}(t) = B^T\exp(A^T(1-t))W^{-1}{\bm{\upkappa}}_f
\end{equation}
Hence, $\kappa(t)$ is obtained through repeated integration of \eqref{eq:interpolator}.
\begin{proposition}
The solution of \eqref{eq:kappa_sys} for $t\in[0,1]$ with $\kappa^{(m+1)}(t)$ given in \eqref{eq:interpolator} in combination with $\kappa(t):=1, \forall t\geq 1$ is a modulating function of order $m$.
\end{proposition}
\begin{proof}
The proof follows directly by identifying $W$ as the controlability Gramian \cite[Definition 2.12, page 240]{gramian} for \eqref{eq:kappa_sys} in the interval $[0,1]$ and $\kappa^{(m+1)}(t)$ obtained from \cite[Corollary 2.14, page 238]{gramian} to reach ${\bm{\upkappa}}(1)={\bm{\upkappa}}_f$ starting from ${\bm{\upkappa}}(0) = 0$, which comply with Definition \ref{def:modulating}.
\end{proof}
As an example, consider $m=1$, such that \eqref{eq:interpolator} lead to
$
\ddot{\kappa}(t) = 6-12t, 
\dot{\kappa}(t) = 6t-6t^2, 
\kappa(t) = 3t^2 - 2t^3
$
for $t\in[0,1]$ in combination with $\kappa(t)=1, \forall t\geq 1$. Hence, $\kappa(t)$ verifies $\kappa(0)=\dot{\kappa}(0)=\dot{\kappa}(1)=0$ and $\kappa(1)=1$.

\subsection{The EDCHO protocol}
\label{ap:edcho}
The EDCHO protocol, presented in \cite{edcho}, is a distributed algorithm designed to compute $\bar{\sigma}_i(t) = (1/\mfs{N})\sum_{i=1}^\mfs{N}\tilde{\sigma}_i(t)$ and its first $m$ derivatives at each agent, where $\tilde{\sigma}_i(t)$ are local time-varying signals. EDCHO can be written as
\begin{equation}
\label{eq:main_algo}
\begin{aligned}
    \dot{\tilde{x}}_{i,\mu} &=  k_\mu \mbox{$\sum_{j=1}^{\mfs{N}}$}a_{ij}\lceil \tilde{y}_{i,0} - \tilde{y}_{j,0} \rfloor^{\frac{m-\mu}{m+1}} + \tilde{x}_{i,\mu+1} \\
    &\text{for }0\leq\mu < m\\
    \dot{\tilde{x}}_{i,m} &= k_m \mbox{$\sum_{j=1}^{\mfs{N}}$}a_{ij}\sgn{\tilde{y}_{i,0} - \tilde{y}_{j,0} }{0} \\
    \tilde{y}_{i,\mu} &= \tilde{\sigma}_i^{(\mu)} - \tilde{x}_{i,\mu}.
\end{aligned}
\end{equation}
where time-dependency is omitted for brevity, $a_{ij}\in\{0,1\}$ are the components of the adjacency matrix of $\mathcal{G}$.
\begin{assumption}
\label{as:bounded}
Given $L>0$, then if follows that
$
\left|\bar{\sigma}^{(m+1)}(t)-\tilde{\sigma}_i^{(m+1)}(t)\right|\leq L, \forall t\geq 0
$.
\end{assumption}
\begin{proposition}\cite[Adapted from Theorem 7]{edcho}
\label{prop:edcho}
Let Assumption \ref{as:bounded} \RevSeven{hold} for given $L$, $\sum_{i=1}^{\mfs{N}} \tilde{x}_{i,\mu}(0) = 0, \forall \mu\in\{0,\dots,m\}$ and a fixed connected communication network $\mathcal{G}$. Then, there exists a time $T>0$ which depends on the initial conditions ${\tilde{y}}_{i,\mu}(0), i\in\mathcal{I}$ and gains $k_0,\dots,k_m>0$ such that \eqref{eq:main_algo} comply $\tilde{y}_{i,\mu}(t)=\bar{\sigma}^{(\mu)}_i(t), \forall t\geq T, \mu\in\{0,\dots,m\}, i\in\mathcal{I}$.
\end{proposition}
Recall that given a value of $L>0$, the design rules for the parameters $k_0,\dots,k_m>0$ are detailed in \cite[Section 6]{edcho}.
\subsection{Proof of Theorem \ref{th:medcho}}
\label{ap:medcho}
First, let $\tilde{x}_{i,\mu}(t)=x_{i,\mu}(t)/\theta$, $\tilde{\sigma}_{i}(t)=\sigma_{i}(t)/\theta$, $ \tilde{y}_{i,\mu}(t)=y_{i,\mu}(t)/\theta$. Hence, using \eqref{eq:medcho}, the dynamics of the transformed variables $\tilde{x}_{i,\mu}(t)$ are equivalent to EDCHO in \eqref{eq:main_algo} from Appendix \ref{ap:edcho}.
Moreover, expanding \eqref{eq:sigma} leads to
$$
\begin{aligned}
&|\sigma_i^{(m+1)}(t)-\sigma_j^{(m+1)}(t)|=\\
&\left|\sum_{\nu=0}^{m+1} \binom{m+1}{\nu} \frac{1}{T_{c}^{m-\nu+1}}\kappa^{(m-\nu+1)}\left(\frac{t}{T_c}\right)\left(s_i^{(\nu)}-s_j^{(\nu)}\right)\right| \\
&\leq \sum_{\nu=0}^{m+1} \binom{m+1}{\nu} \frac{1}{T_{\min}^{m-\nu+1}}K_{m-\nu+1}L_\nu = \theta
\end{aligned}
$$
for all $t\in[0,T_c]$, where Assumption \ref{as:all_bounded} was used. Thus, $|\tilde{\sigma}_i^{(m+1)}(t)-\tilde{\sigma}_j^{(m+1)}(t)|\leq 1$. Hence, Assumption \ref{as:bounded} in Appendix \ref{ap:edcho} is complied for signals $\tilde{\sigma}_1(t),\dots,\tilde{\sigma}_{\mfs{N}}(t)$ with $L=1$ \RevSeven{due to smoothness of both $s_i(t)$ and $\kappa(t)$}. Therefore, Proposition \ref{prop:edcho} in Appendix \ref{ap:edcho} implies that $\tilde{y}_{i,\mu}(t)=(1/{\mfs{N}})\sum_{i=1}^{\mfs{N}} \tilde{{\sigma}}_i(t)$, equivalently $y_{i,\mu}(t) = (1/{\mfs{N}})\sum_{i=1}^{\mfs{N}} {{\sigma}}_i(t), \forall t\geq T$ for some time $T>0$ which depends on the initial conditions ${y}_{i,\mu}(0)$. However, note that since $\sigma_i(0) = \kappa(0)s_i(0)=0,\forall i\in\mathcal{I}$, then $y_{i,\mu}(0) = 0, \forall\mu\in\{0,\dots,m\}$, which means that \eqref{eq:medcho} starts from consensus. Thus, attractivity of \eqref{eq:medcho} along
${y}_{i,\mu}(t) = (1/\mfs{N})\sum_{i=1}^{\mfs{N}}\sigma^{(\mu)}_i(t)
$ 
for all $\mu\in\{0,\dots,m\}$ is maintained through $t\in[0,T_c]$. {Note that $\sigma_i(t)$ is a smooth signal at $t=T_c$ due to the properties of $\kappa(t/T_c)$ and its derivatives at $t=T_c$ given in Definition \ref{def:modulating}. Thus the assumptions of the theorem are maintained for $t\geq T_c$.} Now, $\sigma_i(t)= s_i(t), \forall t\geq T_c$ such that, 
$
{y}_{i,\mu}(t) = \bar{s}^{(\mu)}(t), \ \ \forall t\geq T_c
$ 
and for any $\mu\in\{0,\dots,m\}$.

\subsection{Proof of Theorem \ref{th:ratio}}
\label{ap:ratio}
Assumption \ref{as:bounded} imply that Assumption \ref{as:all_bounded} is complied for each instance of \eqref{eq:medcho} in $\mathcal{M}_{\mf{s}}$. Thus, Theorem \ref{th:medcho} leading to $\mf{y}_{i,\mu}(t)=\bar{\mf{s}}(t), \forall t\geq T_c$. Now, note that $\ell_i^{(\mu)}(t)=0,\forall t\geq 0$ and $\mu>0$ from \eqref{eq:stat_signals}. Hence, $\sigma_i^{\ell}(t) = \kappa(t)\ell_i$ satisfy Assumption \ref{as:all_bounded} with bounds $L_0^{\ell},\dots,L_{m+1}^{\ell}$ given by $L_0^{\ell} = 1, L_\mu=0$ with $\mu>0$. Therefore, the form of the parameter $\theta$ in \eqref{eq:gain} of reduced to $\theta = \frac{K_{m+1}}{T_{\min}^{m-\nu+1}}$. Thus, $l_{i,0}(t)$ converge to $\mfs{N}_{\mfs{L}}/\mfs{N}$ as a consequence of Theorem \ref{th:medcho}. Note that $\mfs{N}_{\mfs{L}}/\mfs{N} \geq 1/\mfs{N}_{\max}$ implying $\max(l_{i,0}(t), 1/\mfs{N}_{\max})=\mfs{N}_{\mfs{L}}/\mfs{N}$ for $t\geq T_c$. Henceforth, $\hat{\mf{p}}_{i,\mu}(t)=\mf{\bar{p}}^{(\mu)}(t)$ is complied for all $t \geq T_c$, satisfying Definition \ref{def:rt_dvo}.

\bibliographystyle{IEEEtran}

\end{document}